\documentclass[reprint,twocolumn,superscriptaddress,showpacs,nofootinbib,notitlepage,preprintnumbers]{revtex4-2}

\usepackage{graphicx}
\usepackage{hyperref}
\usepackage[dvipsnames]{xcolor}
\usepackage{amsmath,amsthm,amssymb}
\usepackage{comment}
\usepackage{microtype}
\usepackage{enumerate}

\DeclareMathOperator{\Tr}{Tr}

\newtheorem*{theorem*}{Theorem}
\newtheorem{lemma}{Lemma}
\newtheorem*{lemma*}{Lemma}

\newtheorem*{problem*}{Problem}

\renewcommand{\paragraph}[1]{\vspace{5.5pt}\noindent\textbf{#1}---\ignorespaces}

\begin{document}
\preprint{LA-UR-24-33001}

\title{Bootstrapping time-evolution in quantum mechanics}

\author{Scott Lawrence}
\email{srlawrence@lanl.gov}
\affiliation{Los Alamos National Laboratory Theoretical Division T-2, Los Alamos, NM 87545, USA}
\author{Brian McPeak}
\email{bmmcpeak@syr.edu}
\affiliation{Syracuse University, Crouse Dr, Syracuse, NY 13210}
\affiliation{McGill University, 3600 Rue University, Montréal, QC H3A 2T8, Canada
}
\author{Duff Neill}
\email{dneill@lanl.gov}
\affiliation{Los Alamos National Laboratory Theoretical Division T-2, Los Alamos, NM 87545, USA}
\date{\today}

\begin{abstract}
	We present a method for obtaining a hierarchy of rigorous bounds on the time-evolution of a quantum mechanical system from an arbitrary initial state, systematically generalizing Mandelstam-Tamm-like relations. For any fixed level in the hierarchy, the bounds are tightest after short time-evolution and gradually loosen over time; we present evidence that for any fixed amount of time-evolution, the bounds can be made arbitrarily tight by moving up in the hierarchy. The computational effort to obtain the bounds scales polynomially with the number of degrees of freedom in the system being simulated. We demonstrate the method on both a single anharmonic oscillator and a system of two coupled anharmonic oscillators.
\end{abstract}

\maketitle

\paragraph{Introduction}Fully non-perturbative methods for studying quantum systems are few and far between. Monte Carlo methods have proven exceptionally powerful, but are unable to probe out-of-equilibrium physics due to a sign problem. In this Letter we use Lagrangian optimization to obtain rigorous and convergent bounds on the time-evolution of quantum systems. These bounds may be thought of as a systematic generalization of Mandelstam-Tamm-like relations~\cite{Mandelstam1991,1972CMaPh..28..251L,PhysRevLett.72.3439,PhysRevLett.103.240501,PhysRevA.106.042436,PhysRevX.12.011038,2024arXiv240904544B}, in the form of a convex optimization program that searches over the space of constraints on a time-dependent observable, arising from uncertainty relations.

%These bounds provide a systematic generalization of the Mandelstam-Tamm relation~\cite{Mandelstam1991}, in the form of a convex optimization program that searches all possible constraints arising from uncertainty relations between any operators related to our observable of interest. The Mandelstam-Tamm relation is a fundamental bound on the rate of change of quantum variables, which has been the subject of much attention over the years in applications to observables in quantum mechanics~\cite{1972CMaPh..28..251L,PhysRevLett.72.3439,PhysRevLett.103.240501,PhysRevA.106.042436,PhysRevX.12.011038,2024arXiv240904544B}.

The method we present is broadly similar to previous bootstrap methods, which have most prominently been used to constrain critical exponents of conformal field theories~\cite{Rattazzi:2008pe,Kos:2016ysd,Simmons-Duffin:2015qma} (see~\cite{Poland:2018epd} for a review). The S-matrix bootstrap program~\cite{chew1961s} was revived in recent years using similar techniques to study S-matrices~\cite{Paulos:2016but, Paulos:2016fap, Paulos:2017fhb} and EFTs~\cite{Caron-Huot:2020cmc,Tolley:2020gtv, Sinha:2020win, Arkani-Hamed:2020blm}. In this Letter our focus is on constraints on quantum mechanical systems. Inspired by the application of positivity to matrix models \cite{Lin:2020mme, Kazakov:2021lel}, these methods have been applied to constraining ground-state physics in quantum mechanics~\cite{Berenstein:2021dyf,Berenstein:2021loy,Nancarrow:2022wdr}, matrix quantum mechanics~\cite{Han:2020bkb,Lin:2023owt,Lin:2024vvg}, lattice systems~\cite{PhysRevLett.108.200404,Anderson:2016rcw,Lawrence:2021msm,Lawrence:2022vsb,Kazakov:2022xuh,Cho:2022lcj, Cho:2023ulr, Kazakov:2024ool, Li:2024wrd}, and systems at finite temperature~\cite{Fawzi:2023fpg,Cho:2024kxn}. Recent work has also explored constraints on hydrodynamics~\cite{heller2023hydrohedron}, and Hilbert space positivity has been combined with lattice data to constrain real-time linear response~\cite{Lawrence:2024hjm}. 

In this Letter we show how to systematically derive time-dependent bounds on expectation values in a quantum mechanical system. Given a set of known expectation values at an initial time $t = 0$, we give an algorithm for obtaining bounds on various expectation values at any later time $T$. The inputs are fairly minimal: we require 
%\textit{positivity}
positivity of the norm of every state in the Hilbert space, some knowledge of the initial state, and the Heisenberg equations of motion. To obtain an algorithm that can be run on a computer of finite size, we make two simplifications: we truncate to a finite set of operators, and then, after passing to the dual optimization problem, we discretize to a finite-dimensional space of functionals. As the truncation and discretization are lifted, meaning that we include larger sets of operators and functionals, the bounds tighten. We present evidence that the bounds can be made arbitrarily tight. 

%Lagrange duality plays a particularly important role in obtaining rigorous bounds. Consider first the ``primal'' problem, which in this case is to minimize an expectation value such as $\langle x \rangle$ at a time $t_*$ subject to the constraint that all norms are positive at every time between $0$ and $t_*$. This requires that we construct an actual function which satisfies the positivity constraints at every point. To make the problem tractable one is required to simplify this space of functions with some sort of discretization. However the true $\langle x(t) \rangle$ may not satisfy our bounds because it may depart from our discretization ansatz -- perhaps drastically. In \textit{practice}, we have actually found in simple examples that discretizing in time does lead to bounds which approach the true values -- it could certainly be useful to develop the primal method further. However in this paper we shall instead take the approach of dualizing the problem and then discretizing in the space of possible functionals -- this appears to be more efficient than the primal approach, but its main virtue, as we shall show, is that it leads to rigorous bounds which approach the true values as both the truncation and discretization are lifted.

\paragraph{Primal problem}
As with other bootstrap-like methods, the heart of our approach is the definition of a particular convex space. Choose a basis of $N$ operators acting on the Hilbert space: $\{\mathcal O_1, \ldots, \mathcal O_N\}$. Our convex space is constructed as a subset of the space of Hermitian matrix-valued functions $M : [0,T] \rightarrow \mathbb C^{N \times N}$. Each such function represents a collection of time-dependent expectation values: $M_{ij}(t) \equiv \langle \mathcal O_i^\dagger(t) \mathcal O_j(t) \rangle$. The constraints on such functions are described below. By convex optimization we can efficiently obtain upper and lower bounds on any linear function of the space; in particular, bounding an expectation value $M_{ij}(T)$.

We will make use of four types of constraints on $M(t)$. The first is an inequality. Because the norm of any state is non-negative, $M(t) \succeq 0$ is positive semi-definite for any $t$. Equivalently, we have $\langle \mathcal O^\dagger \mathcal O\rangle \ge 0$ for all operators $\mathcal O$. 

The remaining constraints are affine equalities. The second constraint comes from the fact that although the $\mathcal O_i$ are linearly independent, the $N^2$ operators $\mathcal O^\dagger_i \mathcal O_j$ are typically not: they may be related by the operator algebra. For each matrix $A_{ij}$ such that $0 = \sum_{ij} A_{ij} \mathcal O_i^\dagger \mathcal O_j$, we impose $\Tr A M(t) = 0$ at all times $t$. The third constraint is our knowledge of expectation values at $t=0$.
We will assume that $M(t=0)$ is known exactly\footnote{Essentially the same method can be used if only some expectation values are known, or even if all that is available are some convex bounds on $M(t=0)$.}. 

The fourth and final constraint is the Heisenberg equations of motion, which relate the time-derivatives of some expectation values to other expectation values:
\begin{equation}
	\frac{d}{dt} \langle \mathcal O(t) \rangle = i \langle [H, \mathcal O] (t) \rangle
	\text.
\end{equation}
Together these three constraints define, for any $N$-operator basis $\{\mathcal O_i\}$, a convex space of functions $[0,T] \rightarrow \mathbb C^{N \times N}$.

Subject to these constraints, we will determine the minimum possible value of some expectation value $\langle \mathcal O_*(T)\rangle$. Assuming this expectation value appears in the matrix $M$, this is equivalent to minimizing $\Tr O M(T)$ for some matrix $O$. This optimization problem may be summarized as follows:
\begin{equation}
 	\label{eq:primalop}
	\begin{split}
		\textbf{minimize } &\Tr O M(T)\\
		\textbf{ subject to }&M(t)\succeq 0\\
		&\Tr A^{(i)} M(t) = 0\\
		&\Tr B^{(j)} M(0) = b_j \\
		&\Tr \left(D^{(k)} - C^{(k)} \frac{d}{dt}\right)M(t) = 0
		\text.
	\end{split}
\end{equation}
Without loss of generality, we will assume for the remainder of this work that the matrices $O,A^\bullet,B^\bullet,C^\bullet$ are all Hermitian, with $b_\bullet$ real.

The minimum of the convex program~\eqref{eq:primalop} corresponds to the minimal value of $\langle \mathcal O(T)\rangle$ obtainable consistent with the constraints detailed above, and thus corresponds to a lower bound on that expectation value. An upper bound can be obtained by re-solving the convex program with an objective of opposite sign. As written, however, this program cannot be solved on a computer of finite size, as the space of functions $M(t)$ is infinite-dimensional.

\paragraph{Anharmonic oscillator}
Before discussing the computational solution of the convex program \eqref{eq:primalop}, let us first illustrate with the example of an anharmonic oscillator. (The example of the \emph{harmonic} oscillator is further explored in Appendix~\ref{app:examples}.) This system's Hamiltonian is 
\begin{equation}
    \label{eq:AHOham}
    \hat H = \frac 1 2 \hat p^2 + \frac{\omega^2}{2} \hat x^2 + \frac{\lambda}{4} \hat x^4
    \text.
\end{equation}
We will take our basis of operators to be $\{\mathcal O_1,\mathcal O_2,\mathcal O_3,\mathcal O_4\} =\{ \hat I, \hat x, \hat p, \hat x^2\} $. 

We may now give concrete examples of the various matrices appearing in the convex program~\eqref{eq:primalop}. To obtain a lower bound on $\langle \hat x(T)\rangle$, we take the objective matrix to be $O_{ij} = \frac{1}{2}( \delta_{1 i} \delta_{2j} + \delta_{1 j} \delta_{2i})$. In this way the matrix $O$ may be said to ``extract'' the expectation of $\hat x$ from the matrix $M$.

Since $\hat x$ is Hermitian, we have an algebraic relation $\hat x - \hat x^\dagger = 0$, implying that $M_{12} = M_{21}$. In the context of the convex program, this constraint may be represented by the choice $A^{(1)}_{ab} = i \delta_{1a}\delta_{2b} - i \delta_{1b} \delta_{2a}$. Since $M$ (being a Hermitian, $4 \times 4$ matrix) lives in a space of $16$ real dimensions, but contains only $11$ linearly independent expectation values, there are $5$ matrices $A^{(i)}$ in total.

Initial values are specified by matrices $B^{(j)}$ (extracting certain expectation values) in combination with real numbers $b_j$. In this work we assume that all expectation values are known at time $t=0$, and therefore $M(0) = M^{(0)}$ is a known matrix. We therefore take a complete basis of $N^2$ Hermitian matrices $B^{(j)}$, with the initial values given by $b_j = \Tr B^{(j)} M_0$.

Finally, by computing commutators with the Hamiltonian, we may derive equations of motion for all elements of $M$. For example we have
$\frac{d \hat x}{dt} = \hat p$,
which may be expressed by defining
$C^{(1)}_{ab} = \frac 1 2 \left(\delta_{a1} \delta_{b2} + \delta_{a2}\delta_{b1}\right)$ and $D^{(1)}_{ab} = \frac 1 2 \left(\delta_{a1} \delta_{b3} + \delta_{a3}\delta_{b1}\right)$. Note that although there is an equation of motion for $\hat 2 \hat H = \hat p^2 + \omega^2 \hat x^2 + \frac{\lambda}{2}\hat x^4$, there is no equation of motion for $\hat p^2$ alone: this would involve operators whose expectation values are not extractable from the $4 \times 4$ matrix $M$.

For many quantum-mechanical systems---including the oscillator examined here and infinite spin systems---there are an infinite number of operators. In order to obtain a problem of a finite size, we must \emph{truncate} to a finite basis of operators; in this case we truncated to a basis of size $N=4$. Bounds obtained from the truncated problem hold true for the original system, as constraints have only been removed, not added. One hopes that by taking a sequence of larger truncations, the lower and upper bounds each converge to the physical answer. As we will see below, this appears to be true in practice.

This truncation is necessary to obtain a finite convex program, but not sufficient. We cannot in practice optimize over all possible functions $\langle x(t) \rangle, \ldots$ because this space is infinite-dimensional. An obvious approach is to discretize time, representing each function by its value at a finite set of points, and replacing the Heisenberg equations of motion by a finite-differencing approximation. However, once these replacements have been made, bounds obtained in this fashion do not necessarily hold for the original problem. Note the contrast with truncation: bounds of the truncated problem may be undesirably loose, but are at least guaranteed to be true. %Naive discretization would require a careful limit to be taken in order to obtain bounds on time evolution (and a skeptical observer might worry that the true limit was not known).

As it turns out, it is possible to discretize time without sacrificing the rigor of the bounds. To do this we must first pass to the \emph{Lagrange dual} of the convex program~\eqref{eq:primalop}.

\paragraph{Dual problem} To understand the dual problem we must first write a Lagrange functional for our convex program~\eqref{eq:primalop} (henceforth termed the \emph{primal} problem). This is given by
\begin{align}\label{eq:lagrange-functional}
	L[M, \Lambda] = \Tr OM(T) - \int_{0}^T \Lambda(t) M(t)\,dt\text.
\end{align}
Here $\Lambda$ is, like $M$, a matrix-valued function of $t$. For any fixed $M$, it is straightforward to evaluate the maximum of $L[M,\Lambda]$ subject to the requirement that $\Lambda$ be positive semi-definite at all times, finding:
\begin{align}
\max_{\Lambda \succeq 0} L(M, \Lambda) \ = \ 
	\begin{cases} 
		\Tr OM & M \succeq 0 \\
		\infty & \text{otherwise} 
	\end{cases}
\end{align}
As a result, the primal optimum may be written as a maximization, followed by a minimization, of the Lagrangian:
\begin{align}
	p^* = \min_{M} \max_{\Lambda \succeq 0} L(M, \Lambda)
	\text.
\end{align}
Above, only the maximization incorporates an explicit positivity constraint. The minimization over $M$ is constrained by the affine relations described above ($\Tr A^{(i)} M = 0$ and so on). The positivity of $M$ is forced by the fact that, for non-positive semi-definite $M$, the inner maximization will be unbounded above.

The \emph{dual problem} is obtained by swapping the order of the two optimizations:
\begin{align}
	d^* =  \max_{\Lambda \succeq 0} \min_{M} L[M, \Lambda]\text.
\end{align}
This problem is of practical use when the inner optimization can be performed analytically---we will accomplish this for our case below. The result is $d^* = \max_{\Lambda \succeq 0} g[\Lambda]$, where $g[\Lambda]$ is the ``Lagrange dual functional''.

In the cases considered in this work, the primal and dual problems have equal optima: $d^* = p^*$. In the general case this need not be true (for the optimization problems in this Letter it follows from Slater's condition~\cite{RePEc:cwl:cwldpp:80}), and all that can be guaranteed without additional assumptions is that $d^* \le p^*$. The stronger result of equality is known as \emph{strong duality}, and the more general inequality is termed \emph{weak duality}.

The dual problem is infinite-dimensional just as the primal problem was. However, weak duality has a useful consequence: for any dual-feasible point, meaning any $\Lambda(t) \succeq 0$, $g[\Lambda]$ is a lower bound on the primal optimum $p^*$. Therefore, we may consider any convenient finite-dimensional space of functions $\Lambda(t)$, and any dual-feasible point in this space will still yield a lower bound on $p^*$. These bounds may be made tighter by lifting the discretization (considering higher-dimensional subspaces), but are rigorous bounds for any fixed discretization.

It remains to compute the Lagrange dual functional $g[\Lambda]$, and to specify a particular discretization.

First note that the Lagrange functional is linear in $M$. Therefore, fixing $\Lambda$, if $L[M,\Lambda]$ has any dependence on $M$, then the minimization over $M$ will be unbounded below. The subsequent maximization over $\Lambda$ requires that we choose $\Lambda$ so that $L$ is $M$-independent. This is where our constraints come in: since various combinations of the matrices $A^\bullet$, $B^\bullet$, $C^\bullet$, and $D^\bullet$ have known traces with $M$, we can form $\Lambda$ out of these matrices and the resulting Lagrangian can still be $M$-independent. We define
\begin{equation}\label{eq:Lambda-form}
	\Lambda(t) =  \lambda_a^{(i)}(t) A^{(i)} + \left(  D^{(k)} + C^{(k)} \frac{d}{dt} \right) \lambda_d^{(k)}(t)
	\text,
\end{equation}
where summation on repeated indices ($i,k$) is implied. This particular form of $\Lambda(t)$ implicitly imposes many constraints on that matrix-valued function; we shall impose one more on $\lambda(t)$, namely that
\begin{equation}\label{eq:lambda-boundary}
	\lambda^{(k)}_d(T) C^{(k)} = O \text.
\end{equation}
Substituting~\eqref{eq:Lambda-form} into the Lagrange functional~\eqref{eq:lagrange-functional}, we evaluate
\begin{equation}
	L = \lambda^{(k)}_d(0) \Tr C^{(k)} M_0 \text,
\end{equation}
which is independent of $M$ as required. Constraints in the primal problem have become degrees of freedom in the dual problem. Both strengthen the bounds---in the dual case, this is because having more degrees of freedom allows us to search over a larger set of functions $\Lambda$.

With this computation complete, we can now write the Lagrange dual problem in the following form:
\begin{equation}\label{eq:dual-full}
	\begin{split}
		\textbf{maximize }&\lambda^{(k)}_d(0) \Tr C^{(k)} M_0\\
		\textbf{subject to }&
		\lambda^{(k)}_d(T) C^{(k)} = O
\\
		&\lambda_a^{(i)}(t) A^{(i)} + \big(  D^{(k)} + C^{(k)} \frac{d}{d t} \big) \lambda_d^{(k)}\succeq 0\text.
	\end{split}
\end{equation}
To obtain a finite-dimensional space of functions $\lambda$ over which to optimize, we take each $\lambda_a^{(i)}$ and $\lambda_d^{(k)}$ to be defined by a quadratic spline with $K$ knots, evenly spaced at times $t_\kappa = \frac{\kappa}{K+1}T$ for integers $1 \le \kappa \le K$. The splines defining $\lambda_a^\bullet$ thus have $K+3$ parameters: three for the initial condition and one at each knot. The splines defining $\lambda^\bullet_d$ are fixed at $T$ by Eq.~(\ref{eq:lambda-boundary}), so there are only $K+2$ parameters for each. Considering $\lambda$ to be a function $\lambda(t,y)$ of these spline parameters $y$, the dual problem above may be written, on this restricted space, in a form finally amenable to numerical optimization:
 \begin{equation}\label{eq:dualop}
	 \begin{split}
		 \textbf{maximize }& \lambda^{(k)}_d(0) \Tr C^{(k)} M_0 \\
	 \textbf{subject to }& \Lambda(t, y) \succeq 0 \text{ (for all $t$)}\text.
	 \end{split}
\end{equation}
For any fixed number of knots, the bounds obtained from solving this problem are looser than those that would in principle be obtained from solving the full dual problem above. As the number of knots is increased, the obtained bounds converge to those that would be obtained from the full (truncated) problem. (See Appendix~\ref{app:discretization} for a proof.)

While the convex program~\eqref{eq:dualop} has a finite number of degrees of freedom, it possesses an infinite number of constraints. This infinite collection of constraints can be imposed in finite time by using an integral over $t$ to construct the barrier function in the interior-point method. This method was used in a similar context in~\cite{Lawrence:2024hjm}, and is described in detail in Appendix~\ref{app:optimization}.

\paragraph{Demonstration} To begin we demonstrate this method to bound $\langle \hat x(T)\rangle$ in the anharmonic oscillator governed by the Hamiltonian~\eqref{eq:AHOham}, taking a coupling $\lambda = 1$. The initial state is defined in terms of eigenstates $|0\rangle,|1\rangle,\ldots$ of the harmonic oscillator (the same Hamiltonian but with $\lambda = 0$):
\begin{equation}
	|\psi_0\rangle = \frac{6}{7}\left[|0\rangle + \frac 1 2 |1\rangle + \frac 1 4 |2\rangle \right]
	\text.
\end{equation}
We use two bases of operators: $\{I, \hat x, \hat p, \hat x^2\}$ (as discussed above) and a larger basis extended by $\{\hat x^3, \hat x^4, \hat p^2, \hat x \hat p, \hat x^2 \hat p\}$. The convex programs thus obtained (see Appendix~\ref{app:construction} for details of the numerical construction) are solved with a standard interior-point method; see~\cite{boyd2004convex} for an introduction and Appendix~\ref{app:optimization} for the details of our implementation.

From these two bases of operators, and varying the number of knots, a sequence of bounds on real-time evolution are obtained. The parameters defining the convex optimization problems behind the bounds are summarized in Table~\ref{tab:one}. Figure~\ref{fig:one} shows the bounds obtained in this fashion, plotted against an exact result obtained by direct diagonalization of the Hamiltonian. Note that two convex programs (a lower bound and an upper bound) must be solved independently for each time $T$ which is plotted.
\begin{table}
	\begin{tabular}{ccccc}
		$N$ & $K$ & Algebraic & Derivatives & Parameters\\\hline
		4 & 0 & 7 & 7 & 35\\
		4 & 3 & 7 & 7 & 77\\
		9 & 3 & 56 & 21 & 441\\
		9 & 6 & 56 & 21 & 672\\
	\end{tabular}
	\caption{\label{tab:one}Parameters characterizing the convex optimization problems used to obtain Figure~\ref{fig:one}. The size of the basis is $N$, and $K$ gives the number of knots. The next two columns specify the number of algebraic constraints (matrices $A$) and derivative constraints (matrices $D$). The final column gives the dimension of the convex optimization problem being solved.}
\end{table}

\begin{figure}
	\centering\includegraphics[width=0.95\linewidth]{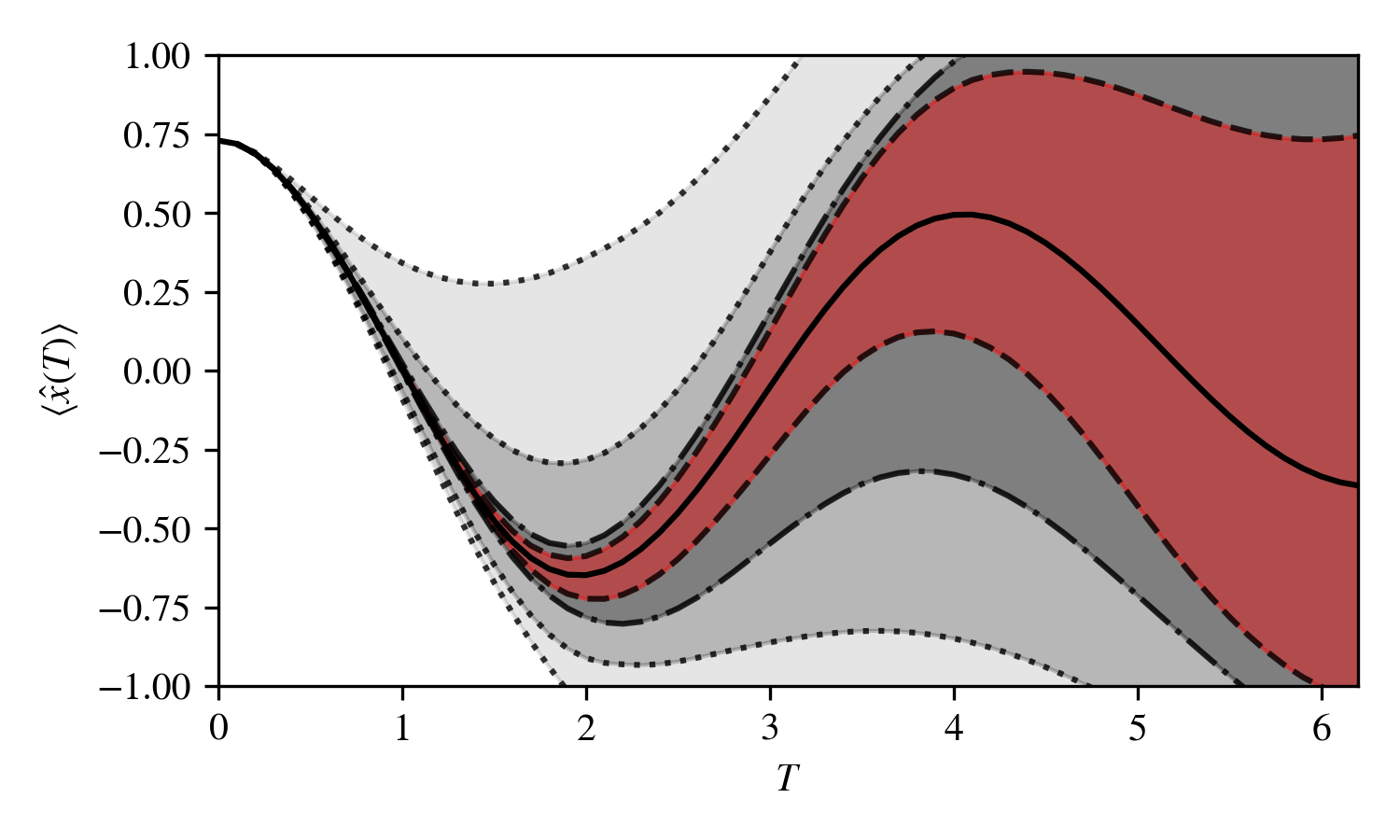}
	\caption{\label{fig:one}Bounds on the time-evolution of the anharmonic oscillator. The solid line shows an exact result obtained by direct diagonalization; the sequentially tighter bounds are detailed in the text and in Table~\ref{tab:one}.}
\end{figure}

The virtue of our approach is that the scaling to larger systems is polynomial without making any approximations, whereas with direct diagonalization (and related methods) the computational cost is exponential in the number of degrees of freedom. Figure~\ref{fig:two} and Table~\ref{tab:two} extend the method to a system of two coupled anharmonic oscillators. The Hamiltonian of this system is
\begin{equation}
	\hat H_2 = \frac{\hat p^2 + \hat q ^2}{2} + \frac {\omega^2} 2 \left(\hat x^2 + \hat y^2\right) + \frac {\left(\hat x - \hat y\right)^2}{2} + \frac{\lambda}{4} \left(\hat x^4 + \hat y^4\right)
	\text,
\end{equation}
where $\hat x,\hat y$ are the position operators of the oscillators, and $\hat p,\hat q$ their conjugate momenta. We adopt $\omega = 1$ and $\lambda = \frac 1 2$. This time, we take as the initial state the tensor product of the ground states of the two uncoupled, unperturbed oscillators; that is, the ground state of the Hamiltonian
\begin{equation}
	\hat H_{2,\text{initial}} = \frac 1 2 \left(\hat p^2 + \hat q^2\right) + \frac{\omega^2}{2} \left(\hat x^2 + \hat y ^2\right)
\end{equation}
Because this system is parity-even, and so $\langle \hat x \rangle = 0$ at all times, we bound the time-dependence of $\langle \hat x^2\rangle$. To define two operator bases, we define a power-counting in which $\hat x,\hat y$ are ``degree 1'' and the conjugate momenta $\hat p,\hat q$ are ``degree 2''. Then the first basis consists of the $8$ monomials which are of degree at most 2, and the second consists of the $26$ monomials of degree at most 4.

\begin{table}
	\begin{tabular}{ccccc}
		$N$ & $K$ & Algebraic & Derivatives & Parameters\\\hline
		8 & 0 & 34 & 18 & 138\\
		8 & 3 & 34 & 18 & 294\\
		26 & 0 & 521 & 103 & 1769
	\end{tabular}
	\caption{\label{tab:two}Parameters characterizing the convex optimization problems used to obtain Figure~\ref{fig:two}. The size of the basis is $N$, and $K$ gives the number of knots. The remaining columns specify the number of algebraic constraints, derivative constraints, and the dimension of the convex program.}
\end{table}

\begin{figure}
	\centering\includegraphics[width=0.95\linewidth]{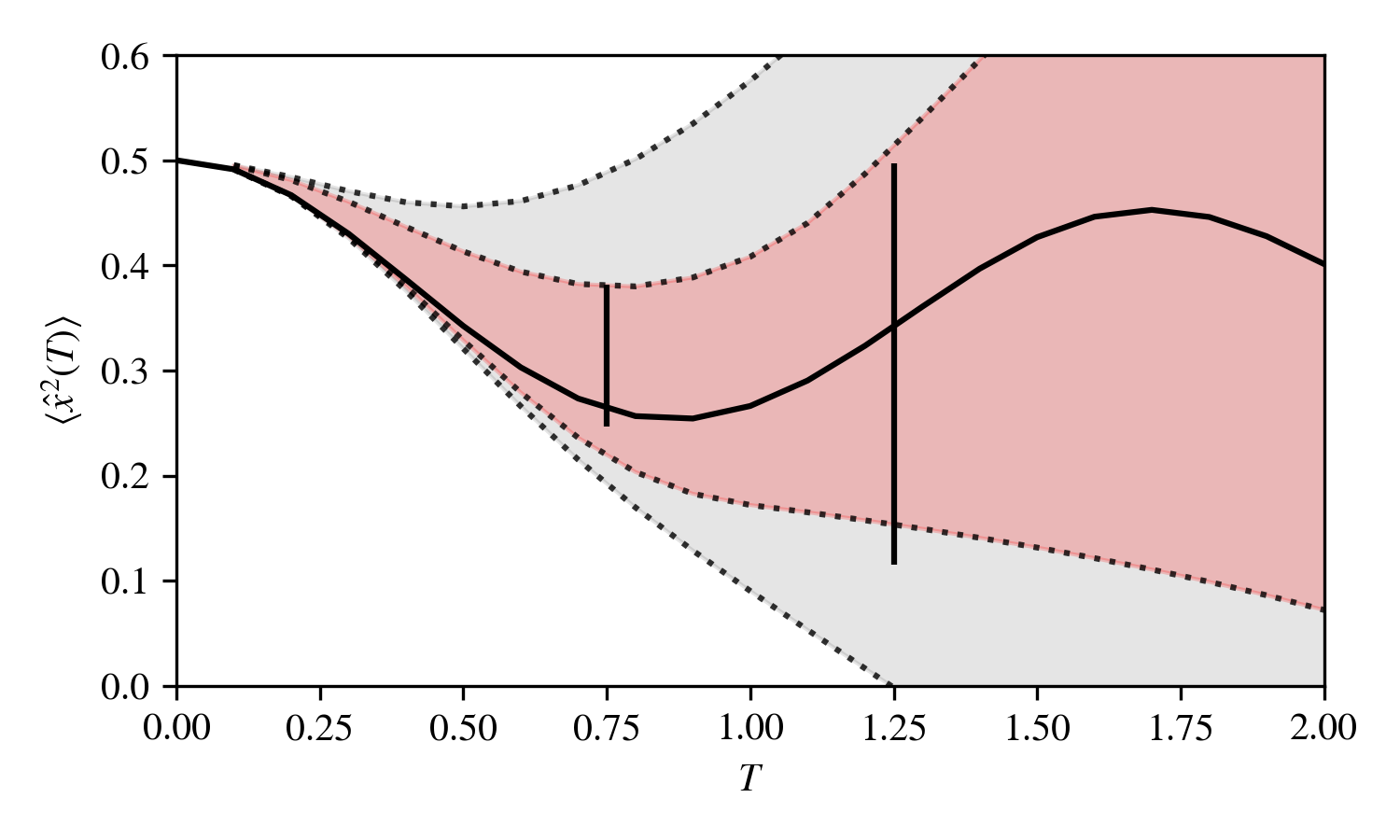}
	\caption{\label{fig:two}Bounds on the time-evolution of two coupled anharmonic oscillators. The solid line shows an exact result obtained by direct diagonalization. The bounds obtained are detailed in Table~\ref{tab:two}; the black bars indicate bounds obtained at $(N,K) = (26,0)$.}
\end{figure}

\paragraph{Discussion}
We have shown how rigorous bounds on the time-evolution of a quantum-mechanical system can be obtained via Lagrange duality and convex optimization. At least at short times, Figures~\ref{fig:one} and~\ref{fig:two} provide strong evidence that the bounds converge to a point as the truncation and discretization are lifted. 

We do not expect that this approach is able to simulate the time-evolution of arbitrary systems in polynomial time (in the number of degrees of freedom). As a result, and because the algorithm is polynomial time in both the number of knots and the size of the basis used, we expect that for a fixed amount of desired time evolution, the size of the basis required to obtain a non-trivial bound must grow exponentially with the number of degrees of freedom being simulated. Demonstrating this scaling in practice is left for future work.

\begin{acknowledgments}
	S.L.~is supported by a Richard P.~Feynman fellowship from the LANL LDRD program. Los Alamos National Laboratory is operated by Triad National Security, LLC, for the National Nuclear Security Administration of U.S. Department of Energy (Contract No. 89233218CNA000001).  B.M.~is supported by the Gloria and Joshua Goldberg Fellowship at Syracuse University and NSERC (Canada), with partial funding from the Mathematical Physics Laboratory of the CRM. D.N.~is supported by the Quantum Science Center (QSC), a National Quantum Information Science Research Center of the U.S. Department of Energy (DOE) and by the U.S. Department of Energy, Office of Science, Office of Nuclear Physics (NP) contract DE-AC52-06NA25396.  
\end{acknowledgments}

\appendix

\section{Optimization}\label{app:optimization}
The core of the method described in this paper is the interior-point algorithm for solving convex programs of the form~\eqref{eq:dualop}. First we will describe the interior-point algorithm in the general case, and then explain how it is specialized to the convex programs solved in this work.

A general convex program with one constraint can be written in the form:
\begin{equation}\label{eq:convex-general}
	\begin{split}
		\textbf{minimize }&f(x)\\
		\textbf{subject to }&h(x) \ge 0\text.
	\end{split}
\end{equation}
Equivalently, we might absorb the constraint into the objective function, defining a new objective
\begin{equation}
	F_0(x) = \begin{cases}f(x) & h(x) \ge 0\\\infty & h(x) < 0\end{cases}
	\text.
\end{equation}
Now an equivalent optimization problem to~\eqref{eq:convex-general} is to perform unconstrained minimization of $F_0(x)$.

To accomplish this we define a sequence of \emph{barrier functions} approximating the constraint, typically by
\begin{equation}
	\phi_s(x) = - s \log h(x)\text.
\end{equation}
These functions have the property that for any $s \ge 0$, $\phi_s(x)$ is finite for feasible points and infinite for infeasible points. We may now define a sequence of objective functions as well:
\begin{equation}
	F_s(x) = f(x) + \phi_s(x)\text.
\end{equation}
Note that this definition is consistent with that of $F_0(x)$ above in the limit $s \rightarrow 0$.

With this family of objective functions defined, the interior-point method is simple: $F_s(x)$ is minimized for successively smaller values of $s$. In practice it is cheap to find a minimum for a large value of $s$ ($s = 10^4$ is used as a starting point in this work). The value of $s$ is then multiplied by some $\mu < 1$ ($\mu = \frac 1 2$ works well in practice and is used here), and the minimization is repeated, using the previous minimum as a starting point.

To perform the intermediate minimizations, we use Newton's method with a Hessian matrix computed exactly (up to floating-point precision). We have found that this outperforms BFGS, which in turn far outperforms simple gradient descent with backtracking. Newton's method is modified to add backtracking (without which it can leave the feasible region).

As noted in the text, the convex program~\eqref{eq:dualop} has an infinite number of constraints. Following~\cite{Lawrence:2024hjm}, we define as our barrier function
\begin{equation}
	\phi_s(y) = -s \int_0^T \log \Lambda(t,y) dt\text.
\end{equation}
This function is integrated numerically, along with its first and second derivatives with respect to the spline parameters $y$, to provide the barrier function as well as evaluations of its gradient and Hessian. The solver used, along with the code to construct the convex program, may be found at~\cite{code}.

\section{Problem construction}\label{app:construction}

Numerical solution of the convex program~\ref{eq:dualop} requires that the matrices $A^\bullet$, $C^\bullet$, and $D^\bullet$ have been computed, so that $\Lambda(t,y)$ may be computed according to Eq.~(\ref{eq:Lambda-form}). We begin this appendix with a summary of the procedure used to find these matrices, and conclude with a discussion of circumstances in which this program has no strictly dual-feasible point. Again, the code used to construct the convex program may be found at~\cite{code}.

The computation of the matrices $A^\bullet$, corresponding to algebraic constraints, is relatively straightforward. The space of Hermitian $N \times N$ matrices has $N^2$ dimensions. Writing each matrix element $M_{ij} \equiv \mathcal O_i^\dagger \mathcal O_j$ in normal-ordered form, we identify all unique monomials that appear in $M$. Each such monomial is associated to an $N\times N$ matrix summarizing the dependence of $M$ on that monomial's expectation value. Label the number of such monomials $m$; then we need only find a basis of $N^2 - m$ matrices for the orthogonal subspace to that spanned by these $m$ matrices. All this is automated with the aid of an appropriate computer algebra system.

Computing $C^\bullet$ and $D^\bullet$ is somewhat more involved. We begin, much as before, by identifying a complete basis of Hermitian operators that appear in $M$. We now wish to identify all derivatives of matrix elements of $M$ that are themselves expressible in terms of matrix elements of $M$. Note, however, that it is possible that while neither $\frac d{dt} M_{22}$ nor $\frac d{dt} M_{33}$ (for example) can be written in terms of matrix elements of $M$, their sum (or some other linear combination) can be.

To deal with this, we label a vector of all expectation values that can be extracted from $M$ by $x$. Now write the Heisenberg equation of motion in the form
\begin{equation}
	\frac{d}{dt} x(t) = \mathfrak d x(t) + \tilde{\mathfrak d} y(t)\text.
\end{equation}
Above, the vector $y$ consists of those expectation values which appear in derivatives of $M$, but cannot be extracted from $M$. We may (numerically, with the aid of a computer algebra system as before) find the matrices $\mathfrak d$ and $\tilde{\mathfrak d}$ in any basis. We desire to find covectors $v$ such that the derivative of $v \cdot x$ may be written exclusively in terms of $M$. In other words, we seek vectors $v$ such that $v^T \tilde{\mathfrak d} = 0$. A complete basis of such vectors $v$ is readily obtained by linear algebra.

With this accomplished, we have identified the complete linear subspace of the Heisenberg equations of motion which can be expressed in terms of our truncated matrix $M$. It is then straightforward to select (arbitrarily) a basis for this subspace, and for each element of that basis construct matrices $(C,D)$ defining the relation.

There is no guarantee that the convex program constructed in this fashion has a strictly dual-feasible point, which is required for the numerical method in the previous section to operate. In practice we find that, depending on the choice of the set of operators used to generate the program, there are often no strictly dual-feasible points. This phenomenon is not specific to the somewhat exotic program constructed here, but instead occurs even with very simple semidefinite programs. A minimal example of this phenomenon is the $2\times 2$ semidefinite program:
\begin{equation}
	\begin{split}
		\textbf{minimize } & x\\
		\textbf{subject to } & \left(\begin{matrix}1 & x\\x & \epsilon\end{matrix}\right) \succeq 0\text.
	\end{split}
\end{equation}
For generic $\epsilon > 0$, the feasible space is of dimension $1$ (and parameterized by $x$). For $\epsilon < 0$, there are of course no feasible points at all. In the special case $\epsilon = 0$, the feasible space is of dimension $0$, consisting of just the point $x=0$. In effect, an extra affine constraint has been ``hidden'' inside the inequality.

This phenomenon causes no major difficulty for the theory of convex optimization---the feasible space is certainly still convex, among other things---but is a serious practical problem for interior-point methods which depend on performing numerical operations on strictly feasible points. In the body of this paper we evaded this issue by taking care to always choose a set of generating operators for which this phenomenon did not occur, \emph{i.e.}~for which there existed a strictly dual-feasible point. In the remainder of this appendix we outline a more general approach.

The interior-point method is unable to function because the parameterizing space of potential solutions is of dimension $m$ (isomorphic to $\mathbb R^m$), while the space of feasible points is in fact of dimension $n < m$. Therefore all points lie on the boundary and no point can be strictly feasible. In order to use an interior-point method, we must first identify the affine constraints which are implied by the inequalities in the convex program. After this is done, these affine constraints may be explicitly added to the formulation of the problem; that is, we may construct an explicit parameterization of the $n$-dimensional subspace which contains all feasible points. Once this is done, interior-point methods may be used immediately.

The challenge, therefore, is to identify the implicit affine constraints. Consider a generic semidefinite program of the form
\begin{equation}
	\begin{split}
		\textbf{minimize }&f(\lambda)\\
		\textbf{subject to }&M + \sum_i \lambda_i L_i \succeq 0
		\text.
	\end{split}
\end{equation}
Here the matrices $M,L_\bullet$ are $N \times N$ and Hermitian. The objective function $f(\cdot)$ is irrelevant to the current question. The positivity constraint is profitably rewritten as a set of infinitely many affine inequalities, namely that for any vector $y \in \mathbb C^N$ we have
\begin{equation}
	y^\dagger M y + \sum_i \lambda_i (y^\dagger L_i y) \ge 0\text.
\end{equation}
Each vector $y$ therefore defines an oriented codimension-$1$ hyperplane. An affine equality is implied by this set of inequalities precisely when the same hyperplane appears twice, but with opposite orientation (thus defining a quantity $a + b\cdot\lambda$ which is required to simultaneously be greater-than-or-equal and less-than-or-equal to zero). It follows that a necessary and sufficient condition is that there exist a pair of vectors $y,\tilde y$ such that
\begin{equation}
	y^\dagger M y = -\tilde y^\dagger M \tilde y
	\text{ and }
	y^\dagger L_i y = -\tilde y^\dagger L_i \tilde y
	\text.
\end{equation}
These are quadratic equalities, readily solved. Of course $y=\tilde y=0$ is always a solution; we are looking for non-vanishing solutions $y,\tilde y$. These can be systematically found by defining expanded matrices $\mathfrak M,\mathfrak L_\bullet$ to be block-diagonal of dimension $2N \times 2N$:
\begin{equation}
	\mathfrak M = \left(\begin{matrix}M & 0 \\ 0 & M\end{matrix}\right) \text{ and }
		\mathfrak L_i = \left(\begin{matrix}L_i & 0 \\ 0 & L_i\end{matrix}\right)\text.
\end{equation}
The desired solutions are now defined by the joint null space of $\mathfrak M,\{\mathfrak L_i\}$. Having been identified, these affine equalities can be made explicit in the semidefinite program, allowing interior-point methods to be used.

\section{Discretization}\label{app:discretization}
This appendix is dedicated to establishing that nothing is lost by using quadratic splines in the definition of $\Lambda(t)$. In particular, let $d^*_\infty$ be the optimum of the original dual problem (\ref{eq:dual-full}), and let $d^*_K$ be the optimum of the discretized dual problem (\ref{eq:dualop}) with $K$ knots. We will show that, assuming there are any strictly dual-feasible points, $\lim_{K\rightarrow\infty} d^*_K = d^*_\infty$.

A sketch of the argument is as follows. Any function $\Lambda(t)$ may be pointwise-approximated arbitrarily well by the ansatz (\ref{eq:Lambda-form}) with $\lambda(t)$ defined by a quadratic spline. Such approximations are not necessarily positive semi-definite; however, as long as $\Lambda$ is \emph{strictly} feasible, and the number of knots in the spline is sufficiently high, one can always find a spline $\lambda(t)$ such that the approximation is itself positive semi-definite (and therefore strictly feasible). The optimal function $\Lambda$, of course, lies on the boundary and is therefore not strictly feasible; however this optimal $\Lambda$ is arbitrarily well approximated (as measured by the objective function) by strictly feasible functions, which are in turn arbitrarily well approximated by splines.

The technical workhorse of this argument is the following lemma on the pointwise convergence of quadratic spline approximations.
\begin{lemma}\label{lemma-spline}
	Let $f(t)$ be a real-valued function on the interval $[0,T]$. Furthermore assume that $f$ and $f'$ are each Lipschitz with constant $C$: $|f(t_1) - f(t_2)| \le C |t_1-t_2|$ and $|f'(t_1) - f'(t_2)| \le C |t_1 - t_2|$. Then for any $\epsilon > 0$ there exists a quadratic spline $\tilde f(t)$, with a finite number of knots $K_\epsilon$, such that at all $t \in [0,T]$, $|f(t) - \tilde f(t)| \le \epsilon$ and $|f'(t) - \tilde f'(t)| \le \epsilon$.
\end{lemma}
\begin{proof}Define $g(t) = \frac{d}{dt} f(t)$. Approximate $g$ by a linear spline with $K$ knots; that is, a function $\tilde g(t)$, linear on intervals $[\frac{T k}{K+1},\frac{T(k+1)}{K+1}]$ and obeying $g(Tk / (K+1)) = \tilde g(Tk/(K+1))$ for integers $k$. From this, define the quadratic spline $\tilde f(t) = f(0) + \int_0^t ds \, \tilde g(s)$. Lipschitz continuity of $g$ implies
	\begin{equation}
		|g(t) - \tilde g(t)| \le \frac{CT}{K+1}\text.
	\end{equation}
	This establishes the convergence of the first derivative. Integrating, we find that
	\begin{equation}
		|\tilde f(t) - f(t)| \le \int_0^t |g(s) - \tilde g(s)| ds \le t \frac{CT}{K+1}\text,
	\end{equation}
	which establishes the pointwise convergence of $\tilde f$ as well. We complete the proof by noting that $K_\epsilon = \epsilon^{-1} C T (1 + T)$ is sufficient.
\end{proof}

We will also use the fact that small perturbations to a matrix---as measured by the Frobenius norm---do not much change the spectrum. This is captured by the following Lemma.
\begin{lemma}\label{frobeigen}
	Let $M$ be a positive definite Hermitian matrix with minimum eigenvalue $\sigma_{\mathrm{min}}$, and let $\Delta$ be a Hermitian matrix of bounded norm $||\Delta||_F \le \sigma_{\mathrm{min}}$. Then the sum is positive semi-definite: $M + \Delta \succeq 0$.
\end{lemma}
\begin{proof}
	We will show that for all (complex) vectors $v$ such that $|v| = 1$, $v^\dagger (M + \Delta) v \ge 0$; this is equivalent to that matrix being positive semi-definite. To this end we provide a lower bound on the first term and an upper bound on the second. Due to the minimum eigenvalue of $M$ we have $v^\dagger M v \ge \sigma_{\mathrm{min}}$. Meanwhile, denoting by $||\Delta||_2$ the spectral norm of $\Delta$, we have $||\Delta||_F \ge ||\Delta||_2 \ge 0$. Therefore $|v^\dagger \Delta v| \le ||\Delta||_F \le \sigma_{\mathrm{min}}$. The desired inequality follows.
\end{proof}

Now we return to considering the problems (\ref{eq:dual-full}) and (\ref{eq:dualop}). First we show that positive-definite spline approximations to strictly feasible points exist, given a sufficient number of knots. Throughout, $||\cdot||_F$ denotes the Frobenius norm.
\begin{lemma}Let $\Lambda(t) \succ 0$ be a strictly feasible point for (\ref{eq:dual-full}). Then for any $\delta > 0$, there exist quadratic splines $(\tilde\lambda_a, \tilde\lambda_d)$ such that $\tilde\Lambda(t)$, defined according to Eq.~(\ref{eq:Lambda-form}), is also strictly feasible, and $||\Lambda(t) - \tilde\Lambda(t)||_F \le \delta$.
\end{lemma}
\begin{proof}By Lemma~\ref{frobeigen}, we may assume without loss of generality that $\delta$ is sufficiently small that for any Hermitian matrix $\Delta$ with $||\Delta||_F < \delta$, $\Lambda(t) - \Delta \succ 0$ at all times $t$. 
	As $\Lambda(t)$ is feasible, there exist $(\lambda_a,\lambda_d)$ which generate $\Lambda(t)$ according to Eq.~(\ref{eq:Lambda-form}). By Lemma~\ref{lemma-spline}, those functions $\lambda_a,\lambda_d$ can be approximated by splines such that
	\begin{equation}
		\begin{split}
			|\tilde \lambda_\bullet(t) - \lambda_\bullet(t)| &\le \epsilon\text{, and}\\
			|\tilde \lambda_\bullet'(t) - \lambda_\bullet'(t)| &\le \epsilon\text.
		\end{split}
	\end{equation}
	Label by $X$ the norm of the maximum matrix element of all matrices $A$,$C$,$D$. Then on inspection of Eq.~(\ref{eq:Lambda-form}), the above implies
	\begin{equation}
		||\Lambda(t) - \tilde\Lambda(t)||_F \le 3 X \epsilon
	\end{equation}
	We complete the proof by taking $\epsilon = \frac 1 3 X^{-1} \delta$.
\end{proof}

Finally, note that, since the space of feasible functions $\Lambda(t)$ is convex and (by assumption) has an interior, there exist strictly feasible points $\Lambda(t)$ that are arbitrarily close to the optimal $\Lambda^*(t)$ (which is a point on the boundary, and therefore not itself strictly feasible). By the above Lemma these strictly feasible $\Lambda(t)$ are in turn arbitrarily well approximated by spline approximants $\tilde \Lambda(t)$.

Critically, in this appendix we have held the truncation of operators fixed (corresponding to the dimension $N$ of the matrices from which the convex program is built). We are not currently able to prove that, in the limit as this truncation is lifted, the upper and lower bounds converge to the same point. This remains an important topic for future work.

\begin{widetext}
	\section{An example}
\label{app:examples}

It is useful to work out some of the details for the anharmonic oscillator with $\lambda = 0$, making it the \emph{harmonic oscillator}. Taking also $\omega = 1$, this is described by $\hat{H} = \hat{p}^2 / 2 + \hat{x}^2 / 2$. We will restrict to the $3 \times 3$ matrix generated by the set $\{ 1, \hat{x},  \hat{p} \}$. The matrix $M$ becomes
\begin{align}
	\begin{pmatrix} 1 & \langle \hat{x} \rangle & \langle \hat{p} \rangle \\
				\langle \hat{x} \rangle & \langle \hat{x}^2 \rangle & \langle \hat{x}\hat{p} \rangle \\
				\langle \hat{p} \rangle & \langle \hat{x}\hat{p} \rangle - i & \langle \hat{p}^2 \rangle
	\end{pmatrix}
\end{align}
From the form of $M$, we can see that the algebraic constraints will be 
\begin{equation}
	A^{(1)} = 	\begin{pmatrix} 0 & i & 0 \\
				 -i & 0 & 0 \\
				  0 & 0 & 0
	\end{pmatrix}  , \quad
	 A^{(2)} = 	\begin{pmatrix} 0 & 0 & i \\
				 0 & 0 & 0 \\
				  -i & 0 & 0
	\end{pmatrix} , \quad
	 A^{(3)} = 	\begin{pmatrix} 1 & 0 & 0 \\
				 0 & 0 & -i \\
				 0& i & 0
	\end{pmatrix}\text.
\end{equation}
There are six derivative constraints, coming from the $C$s 
\begin{equation}
	C = \left\{	\begin{pmatrix} 1 & 0 & 0 \\
				 0 & 0 & 0 \\
				  0 & 0 & 0
	\end{pmatrix} \, ,
	\begin{pmatrix} 0 & 1 & 0 \\
				 1 & 0 & 0 \\
				  0 & 0 & 0
	\end{pmatrix} \, ,
	\begin{pmatrix} 0 & 0 & 1 \\
				 0 & 0 & 0 \\
				 1 & 0 & 0
	\end{pmatrix} \, ,
	\begin{pmatrix} 0 & 0 & 0 \\
				 0 & 1 & 0 \\
				 0 & 0 & 1
	\end{pmatrix} \, ,
	\begin{pmatrix} 0 & 0 & 0 \\
				 0 & 1 & 0 \\
				 0 & 0 & -1
	\end{pmatrix} \, ,
	\begin{pmatrix} 0 & 0 & 0 \\
				 0 & 0 & 1 \\
				 0 & 1 & 0
	\end{pmatrix}\right\} .
\end{equation}
with corresponding $D$s
\begin{equation}
	D = \left\{	\begin{pmatrix} 0 & 0 & 0 \\
				 0 & 0 & 0 \\
				  0 & 0 & 0
	\end{pmatrix} \, ,
	\begin{pmatrix} 0 & 0 & 1 \\
				 0 & 0 & 0 \\
				  1 & 0 & 0
	\end{pmatrix} \, ,
	\begin{pmatrix} 0 & -1 & 0 \\
				 -1 & 0 & 0 \\
				 0 & 0 & 0
	\end{pmatrix} \, ,
	\begin{pmatrix} 0 & 0 & 0 \\
				 0 & 0 & 0 \\
				 0 & 0 & 0
	\end{pmatrix} \, ,
	\begin{pmatrix} 0 & 0 & 0 \\
				 0 & 0 & 2  \\
				 0 & 2  & 0
	\end{pmatrix} \, ,
	\begin{pmatrix} 0 & 0 & 0 \\
				 0 & -2 & 0  \\
				 0 & 0 & 2
	 \end{pmatrix} \,\right\} .
\end{equation}
\end{widetext}
In particular, because the system closes (commutators with $\hat{H}$ do not increase the length of strings of $\hat{x}$s and $\hat{p}$s), we find a derivative constraint for every independent element of $M$.

Now let us derive a simple bound on $\langle \hat{x} \rangle$ due to energy conservation: positivity of the upper left $2\times 2$ minor of $M$ implies that
\begin{align}
	\langle \hat{x}(t) \rangle^2 \leq \langle \hat{x}(t)^2 \rangle \, .
\end{align}
which applies for all times. Since $\hat{p}^2$ is positive, we have $\hat{x}^2 = 2E - \hat{p}^2>  2E$. The result is the bound
\begin{align}
	| \langle \hat{x} (t)  \rangle | \leq \sqrt{2 E} \, .
\end{align}
We would like to prove this bound using the dual formulation, \textit{i.e.}\ to derive a functional that proves this bound. Recall that every $\Lambda(t) \succeq 0$ which leads to a finite Lagrangian implies a bound. It turns out that the $\lambda_a$ functions will not be needed, and a single $\delta$-function will be strong enough to prove this (non-optimal!) bound. Our bound does not use any initial expectation values except the energy, so it will only require $C^{(1)}$ and $C^{(4)}$. The functional becomes
\begin{align}
	\Lambda(t) = \delta(t - T) 
	 \begin{pmatrix} - \lambda^{(1)}_d & 1/2 & 0 \\
	   1/2 &   -  \lambda^{(4)}_d & 0 \\
	   0 & 0 &  - \lambda^{(4)}_d 
	\end{pmatrix}
	\label{eq:HOlambda}
\end{align}
With this $\Lambda$, the Lagrangian this evaluates to 
\begin{align}
	 L =\lambda^{(1)}_d  + 2 E \lambda^{(4)}_d  \, .
\end{align}
Maximizing $L$ such that the matrix appearing in equation~\eqref{eq:HOlambda} is positive yields the lower bound
\begin{align}
 d^* = - \sqrt{2 E} \, .
 \end{align}
Switching the sign on $O$ yields an upper bound, and in the end we find
\begin{align}
 - \sqrt{2 E} \leq \langle \hat{x}(t) \rangle \leq \sqrt{2E}
\end{align}
as expected. 

\bibliography{refs}
\end{document}